\documentclass[12pt,a4paper]{amsart}

\usepackage[latin1]{inputenc}     
\usepackage{amsmath,amssymb,amsthm}
\usepackage{latexsym}
\usepackage{amsfonts}
\usepackage{bbm,bm,dsfont} 
\usepackage{color} 
\usepackage{graphicx}

\numberwithin{equation}{section} 

\theoremstyle{definition}
\newtheorem{proposition}{Proposition}

\newtheorem{remark}{Remark}
\newtheorem{theorem}{Theorem}
\newtheorem{lemma}{Lemma}
\newtheorem{corollary}{Corollary}

\newcommand{\cal}{\mathcal}

\newcommand{\sfe}{\mathsf{E}}

\newcommand{\sfm}{\mathsf{M}}
\newcommand{\sfp}{\mathsf{P}}

\newcommand{\N}{\mathbb N} 
\newcommand{\R}{\mathbb R} 
\newcommand{\Z}{\mathbb Z} 
\newcommand{\C}{\mathbb C} 
\newcommand{\T}{\mathbb T} 

\newcommand{\fii}{\varphi} 

\newcommand{\hi}{\mathcal{H}} 
\newcommand{\hil}{{\mathcal H}} 
\newcommand{\ki}{\mathcal{K}} 
\newcommand{\li}{\mathcal{L}} 

\newcommand{\lh}{\mathcal{L(H)}} 
\renewcommand{\th}{\mathcal{T(H)}} 
\newcommand{\eh}{\mathcal{E(H)}} 

\newcommand{\tr}[1]{\mathrm{tr}\left[#1\right]} 
\def\<{\langle} 
\def\>{\rangle} 
\newcommand{\ket}[1]{|#1\rangle} 
\newcommand{\kb}[2]{|#1 \rangle\langle #2|} 
\newcommand{\ip}[2]{\left\langle #1 | #2 \right\rangle} 
\newcommand{\no}[1]{\left\|#1\right\|} 

\newcommand{\Eo}{\mathsf{E}} 
\newcommand{\Fo}{\mathsf{F}} 
\newcommand{\Go}{\mathsf{G}}
\newcommand{\Qo}{\mathsf{Q}} 
\newcommand{\Po}{\mathsf{P}} 
\newcommand{\Mo}{\mathsf{M}} 
\newcommand{\No}{\mathsf{N}} 

\newcommand{\bo}[1]{\mathcal{B}\left(#1\right)} 
\def\d{{\mathrm d}} 

\begin{document}

\title[Number and phase]{Number and phase:\\ {\tiny complementarity and joint measurement uncertainties}}

\author{Pekka Lahti}
\address{Turku Centre for Quantum Physics, Department of Physics and Astronomy, University of Turku, FI-20014 Turku, Finland}
\email{pekka.lahti@utu.fi}
\author{Juha-Pekka Pellonpää}
\address{Turku Centre for Quantum Physics, Department of Physics and Astronomy, University of Turku, FI-20014 Turku, Finland}
\email{juhpello@utu.fi}
\author{Jussi Schultz}
\address{Turku Centre for Quantum Physics, Department of Physics and Astronomy, University of Turku, FI-20014 Turku, Finland}
\email{jussi.schultz@gmail.com}

\begin{abstract}
We show that number and canonical phase (of a single mode optical field)  are complementary observables. We also bound the measurement uncertainty region for their approximate joint measurements.

\end{abstract}

\maketitle


\section{Introduction}

Analogously to position and momentum of a quantum object,
number and phase of a single mode optical field are often considered as  an example of a pair of observables which is complementary and for which the uncertainty relations put  severe limitations both for preparations and measurements. 
However, since there is no phase shift covariant spectral measure solution to the quantum phase problem it has remained a challenge 
to formulate the exact content of these intuitive ideas for this pair of observables. 

The notion of complementarity, which goes back to the 1927 Como lecture of Niels Bohr \cite{Bohr1928} and which was strongly advocated also by Wolfgang Pauli \cite{Pauli1933}, is often discussed only rather vaguely  and mostly in connection with Werner Heisenberg's uncertainty relations \cite{Heisenberg1927}. However, the notion of mutual exclusiveness which is associated with the idea of  complementarity has rather straightforward independent formulations in quantum mechanics, and,  like uncertainty, it has both probabilistic and measurement theoretical aspects. Along with Bohr \cite{Bohr1935}, we say that two observables are complementary if all the instruments (measurements) which allow their unambiguous definitions are mutually exclusive. The notion of mutual exclusiveness of measurements is easily expressed with respect to the order structure of the set of quantum effects, sharp or unsharp. Following Pauli \cite{Pauli1933}, one may also say that two observables are probabilistically complementary if certain predictions concerning the measurement outcomes of these observables are mutually exclusive. In addition, with the notion of value complementarity of two observables one often refers to the case where sharply defined  value (exact knowledge) of one observable implies uniform distribution (complete ignorance) on the values of the other observable. These notions have obvious expressions in terms of the measurement outcome probabilities of quantum mechanics.
Straightforward formulations of the three versions of complementarity  have been proposed and studied, for instance, in  \cite{Kraus1987,OQP}.

Concerning number and phase, it is, perhaps, well known that they are probabilistically complementary as well as value complemenary, see, for instance \cite[Proposition 16.2 and 16.3]{Kirja},  but it has remained an open question if among the phase shift covariant phase observables there is any which would be complementary with the number \cite{BLPY2001}. This question is now settled in Section \ref{Cchapt} where it is shown that the canonical phase and number form a complementary pair.

Complementary observables are necessarily incompatible, that is, they cannot be measured jointly. This leads one to study their approximate joint measurements, a topic which has gained a substantial clarification in recent years. 
Rather than digging in the extensive history of  the topic, we refer to  the relevant chapters of the monograph \cite{Kirja}. 
In Section \ref{MURchapt} we  follow the ideas and methods initiated in  \cite{Ludwig1983,RFW2004} and further developed, for instance, in
\cite{BLWJMP2014,BLW2014,BKWerner2016}, to bound the measurement uncertainty domain of the complementary pair of number and canonical phase. 

Throughout the paper we use freely the standard notions and terminology of Hilbert space  quantum mechanics. 
Yet, we start with a short account of the main terminology and the basic results concerning the canonical phase observable.

\section{Basic notions}

Let $\hil$ be a Hilbert space, $\{\vert n \rangle \mid n\in\N\}$ an orthonormal basis of $\hil$, and $N=\sum_{n=0}^\infty n\vert n\rangle\langle n\vert$ the corresponding number operator. 
Let $\lh$ and $\th$ denote, respectively,  the sets of bounded and trace class operators on $\hi$.
We also let $\mathcal S(\hi)\subset\th$ denote the set of positive, trace one operators (states).
We denote by $\No:2^{\N}\to\lh$ the spectral measure of $N$ and call it the {\em number observable}. With any observable, like $\No$, we let $\No_\rho$ denote the probability measure $Y\mapsto\No_\rho(Y)=\tr{\rho\No(Y)}$ defined by the observable and a state $\rho\in\mathcal S(\hi)$.

Let $\bo{[0,2\pi)}$ be the Borel sigma algebra of $[0,2\pi)$.
By a  {\em phase observable} we mean any normalized positive operator measure (semispectral measure)  $\sfe:\bo{[0,2\pi)}\to\lh$ which is
covariant under the phase shifts generated by the number observable, that is, satisfies the 
 condition $e^{i\theta N}\sfe(X)e^{-i\theta N} = \sfe(X\dot{+}\theta)$ for all $\theta\in[0,2\pi)$ and $X\in\bo{[0,2\pi)}$, where $\dot{+}$ denotes addition modulo $2\pi$.
The structure of such observables is completely known, see, for instance, \cite{Holevo,JPVK,Kirja}.
Among them there is the one referred to
the {\em canonical} phase observable, which we denote  by  $\Phi:\bo{[0,2\pi)}\to\lh$ and which has the effects
\begin{equation}\label{phase}
\Phi(X) = \sum_{m,n=0}^\infty \int_X  e^{i(m-n)\theta}\, \frac{d\theta}{2\pi} \vert m \rangle\langle n\vert.
\end{equation}

There are several properties which distinguish $\Phi$ as the {\em canonical} phase  among all the phase observables $\sfe$. Without entering the whole list of such properties,\footnote{A reader interested in those properties of $\Phi$  may 
check the list of 19 items of \cite[Sect.\ 4.8]{JPVK} together with some further properties  \cite{HePe,JPJussi}.
 } we mention here only the fact that, up to unitary equivalence, the canonical phase is the only phase observable which generates number shifts: 
 $V^{(k)}\No(Y+k)(V^{(k)})^*=\No(Y)$, where $V^{(k)}=\int_0^{2\pi}e^{ik\theta}\,d\Phi(\theta)$ are the cyclic moment operators of $\Phi$.  
We recall also that the spectrum of the effect $\Phi(X)$, $0\ne\Phi(X)\ne I$,  is the whole interval $[0,1]$ with no eigenvalues. 
In particular, for any $\theta\in[0,2\pi)$ and for any $\epsilon>0$, the (operator) norm of the effect $\Phi\big((\theta-\epsilon,\theta+\epsilon)\cap[0,2\pi)\big)$ equals one.
Thus, for each point $\theta\in[0,2\pi)$ there is a sequence of unit vectors $(\psi_i)_{i\in\N}$ such that the probability measures $X\mapsto\ip{\psi_i}{\Phi(X)\psi_i}$ tend, with increasing $i$, to the point measure $\delta_\theta$ at $\theta$.
In such a case, the number probabilities $|\ip{\psi_i}{n}|^2$ tend to zero for all $n$.
 Observing, in addition, that in the number states $\ket n$ the phase distribution is uniform,
$\ip{n}{\Phi(X)|n}=\int_X\frac{d\theta}{2\pi}=\ell(X)$, 
 the probabilistic and the value complementarity of the pair $(\No,\Phi)$ become obvious.

As well-known, number $\No$ and phase $\Phi$ are incompatible observables, that is, they cannot be measured jointly. Indeed, since $\No$ is a spectral measure,
their joint measurement $\Mo$ would necessarily be of the product form, that is,   $\sfm(n,X)=\kb nn\Phi(X)=\Phi(X)\kb nn$ for any $n\in\N,X\in\bo{[0,2\pi)}$
(see, for instance, \cite[Proposition 4.8]{Kirja}). But this would imply that $\Phi(X)=\ell(X)\,I$, which contradicts \eqref{phase}.

Though  $\Phi$ and $\No$ have no joint observable, there are observables $\Mo:\bo{[0,2\pi)\times\N}\to\lh$ having either $\Phi$ or $\No$ as a margin, that is, 
either $\Mo_1=\Phi$, with $\Mo_1(X)=\Mo(X\times\N)$, or $\Mo_2=\No$, with
$\Mo_2(Y)=\Mo([0,2\pi)\times Y)$. In either case the joint observable is a smearing of the exact margin.
Indeed, if $\No=\Mo_2$, then $\Mo(X\times Y)= \Mo_1(X)\No(Y)$ (cf.\ above) and each $\Mo_1(X)$ is a function of $N$ so that
 $\Mo(X\times Y)=\sum_{n\in Y} p(X,n)\kb nn$, with a Markov  kernel $\bo{[0,2\pi)}\times\N\ni(X,n)\mapsto p(X,n)\in[0,1]$.
On the other hand,
if $\Mo_1=\Phi$, then again  there is a kernel $p:[0,2\pi)\times 2^{\N}\to[0,1]$ such that $\Mo$ is obtained as
 \begin{equation}\label{tarkkavaihe}
 \Mo(X\times Y)=\int_X p(\theta,Y)\, d\Phi(\theta),\footnote{We recall that this integral simply means that   for each state $\rho$, 
 $\tr{\rho\Mo(X\times Y)}=\int_Xp(\theta,Y)\,d\Phi_\rho(\theta)$, the integral of the (measurable) function $\theta\mapsto p(\theta,Y)$ with respect to the probability measure $\Phi_\rho$.}
 \end{equation}
so that, in particular,  for each $Y\in 2^{\N}$, $\Mo_2(Y)=\int_0^{2\pi}\,p(\theta,Y)\,d\Phi(\theta)$.
The structural similarity of the two cases is due to the fact that
 both $\Phi$ and $\No$ are rank-1 observables, for details, see \cite{JPJussi,JP2014}.

The above results contain also the following well-known facts.  In any of the sequential measurements  (in either order), 
if the first measurement is exact, that is, measures either $\No$ or $\Phi$, then any information 
on the other observable
coded in the initial state of the measured system is completely lost in the following precise sense: if, say, $\No$ is measured first, with an instrument $\mathcal I$, in a state $\rho$, then the subsequent phase probabilities  are $\tr{\mathcal I(\N)(\rho)\Phi(X)}= \tr{\rho\mathcal I(\N)^*(\Phi(X))},$  where the `distorted' phase effects $\mathcal I(\N)^*(\Phi(X))$ 
are  smearings of the number observable for some kernel $n\mapsto p(X,n)$. Similarly, if one first performs an exact phase measurement, with an instrument  $\mathcal J$, say, then the subsequent number probabilities are $\tr{\mathcal J([0,2\pi))(\rho)\kb nn}= \tr{\rho\mathcal J([0,2\pi))^*(\kb nn)},$  where the `distorted' number effects 
$\mathcal J([0,2\pi))^*(\kb nn)$ are smearings of the phase observable $\Phi$ with a kernel  $\theta\mapsto p(\theta,\{n\})$.

We now turn to study the complementarity of the number and the canonical phase.

\section{Complementarity of the pair $(\No,\Phi)$}\label{Cchapt}

As already pointed out,
the pair $(\No,\Phi)$ is known to be both probabilistically complementary  and value complementary, but it has remained an open question if they are also complementary.   
This question will now be settled with Theorem \ref{corollary} which shows that for each finite subset $Y\subset\N$ and $X\in\bo{[0,2\pi)}$, for which $\Phi(X)\ne I$, the greatest lower bound of the effects $\Phi(X)$ and $\No(Y)$ exists in the  partially ordered set of effects $\eh=\{E\in\lh\,|\, 0\leq E\leq I\}$ and equals the null effect, that is
\begin{equation}\label{komplementaarisuus}
 \Phi(X)\wedge\No(Y) =0.
\end{equation}
It is this relation which we take to express  the complementarity of the pair $(\No,\Phi)$ in the sense  that all the measurements which serve to define these observables are mutually exclusive.
In fact, if \eqref{komplementaarisuus} were not true, then for  some such $X$ and $Y$ there would be
an effect $E$ below both $\No(Y)$ and $\Phi(X)$, so that, in any state $\rho$,
the probability $\tr{\rho E}$  would also be a common lower bound for the corresponding number and the phase probabilities. Thus, with measuring the effect $E$ in any state  one would also get information from the effects $\No(Y)$ and $\Phi(X)$ in that state.   
Relation \eqref{komplementaarisuus} excludes such measurements.

 The order structure of the set of effects is known to be quite complicated when compared with the order structure of the set of projections. However,
a characterization of pairs of effects $E,\,F\in\eh$ for which $E\land F$ exists  has been obtained   \cite{DUetal_2006}, and, in particular,  
it is known that if one of them is a projection then their greatest lower bound always exists  \cite[Corollary  3.1]{DUetal_2006}.
Therefore, $\Phi(X)\land\No(Y)$ exists for any $X\in\bo{\T}$ and $Y\subset\N$, and it remains to be shown that all these meets are zero whenever $Y$ is a finite set and $X$ such that $\ell(X)<1$ (i.e.\ $\Phi(X)\ne I$). Clearly, such a result depends on the explicit properties of the number and the canonical phase.

From now on we identify the phase interval $[0,2\pi)$ (addition modulo $2\pi$) with the torus $\T$ in the usual way through the map $\theta\mapsto e^{i\theta}$, 
denoting still by $d\ell(\theta)=\frac{d\theta}{2\pi}$ the normalized   
measure on  $\T$. 
Let $\Qo$ be the canonical spectral measure of the Hilbert space $\widetilde\hi=L^2(\T)$ and let $\{e_k\,|\, k\in\Z\}$ be its Fourier basis,  that is, 
$e_k(\theta)=e^{-ik\theta}$. 
Let $P_\N$ be the projection $\sum_{n=0}^\infty\kb{e_n}{e_n}$. The Naimark projection of $\Qo$ onto $P_\N(\widetilde\hi)$, that is, the map
 $X\mapsto P_\N\Qo(X)|_{P_\N(\widetilde\hi)}$ is exactly of the form \eqref{phase}. In fact, $\Qo$ is the minimal Naimark dilation of $\Phi$ \cite[Theorem 8.1]{Kirja}. 

We identify $\hi$ with the subspace  $P_\N(\widetilde{\hi})$ of  $\widetilde\hi$   
via the isometry $V:\ket n\mapsto e_n,$
so that $P_\N=VV^*$ and
 $$
 \Phi(X)=V^*\Qo(X)V=V^*P_\N\Qo(X)V$$ 
 for all $X\in\bo{\T}$.

\begin{remark} 
Let  $\sfp$ be the spectral measure with the (atomic) projections $\kb{e_k}{e_k},$ $k\in\Z$.
In   \cite[Example 4.2]{PekkaKari} it was shown that the pair $(\Qo,\sfp)$  of $L^2(\T)$ is complementary, that is, 
 $\Qo(X)\wedge\sfp(Y)=0$ 
 for all $X\in\bo{\T}$, for which $\Qo(X)\ne I_{\widetilde\hi},$
 and for all finite $Y\subset\Z$.  
 The corresponding result for the position-momentum pair  $(\Qo,\sfp)$ of $L^2(\R)$ is well known, see, e.g., \cite[Proposition 8.2]{Kirja}. 
Though $\Phi(X)=V^*\Qo(X) V$ and $\No(Y)=V^*\Po(Y) V$ (= $\Po(Y)$, for $Y\subset\N$), the noncommutativity of $P_\N$ and $\Qo(X)$ prevents one to conclude the disjointness of the effects $\Phi(X)$ and $\No(Y)$ directly from 
the disjointness of the projections  $\Qo(X)$ and $\Po(Y)$.
\end{remark}

\begin{lemma}\label{lemma}
Let $\alpha\ge 0$ and $X\in\bo{\T}$ such that $\Phi(X)\ne I$.
Then $\alpha\kb{0}{0}\le \Phi(X)$ implies $\alpha=0$.
\end{lemma}

\begin{proof} 
Suppose that 
$\alpha\kb{0}{0}\le \Phi(X)=V^*\Qo(X)V,$   
that is, $\alpha\kb{e_0}{e_0}\le V\Phi(X)V^*=P_\N\Qo(X)P_\N = [\Qo(X)P_\N]^*[\Qo(X)P_\N]$
and note that $\Phi(X)\ne I$ if and only if $\Qo(X)\ne I_{\widetilde\hi}$ if and only if $\ell(X)<1$.
Let $\ki=\overline{\Qo(X)P_\N(\widetilde{\hi})}$. 
Define an operator $D\in\li\big(\widetilde\hi\big)$ by
$D\big(\Qo(X)\psi\big)=\sqrt{\alpha}\<e_0|\psi\>e_0$, $\psi\in P_\N(\widetilde{\hi})$, and $D\fii=0$, $\fii\in\ki^\perp$. Indeed, $D$ is clearly linear and well defined since, if
$\Qo(X)\psi=\Qo(X)\psi'$, $\psi,\,\psi'\in P_\N(\widetilde{\hi})$, i.e.\ $\Qo(X)\psi_{-}=0$, $\psi_-=\psi-\psi'$, then
$$
0\le \|D(\Qo(X)\psi_-)\|^2=\<\psi_-|\alpha e_0\>\<e_0|\psi_-\>
\le\<\psi_-|P_\N\Qo(X)P_\N\psi_-\>
=\<\psi_-|\Qo(X)\psi_-\> =0
$$
so that $D(\Qo(X)\psi)=D(\Qo(X)\psi')$. Similarly, $\|D(\Qo(X)\psi)\|\le\|\Qo(X)\psi\|$, $\psi\in P_\N(\widetilde{\hi})$, showing that $D$ is bounded and thus extends 
to the whole $\widetilde\hi$. Since the range of $D$ is $\C e_0$, one has   
$D=\kb{e_0}{\eta}$ for some $\eta\in \widetilde\hi$.
In addition, since $D\Qo(X)P_\N=\sqrt{\alpha}\kb{e_0}{e_0}$,
$$
\alpha\kb{e_0}{e_0}=[\Qo(X)P_\N]^*D^*D[\Qo(X)P_\N]=\kb{\eta'}{\eta'}
$$
where $\eta'=  
P_\N\Qo(X)\eta$ and also $\eta'=z\sqrt{\alpha}e_0$, $z\in\T$.
Now
$\<e_m|\Qo(X)\eta\>=
\<e_m|P_\N\Qo(X)\eta\>=
\<e_m|\eta'\>=0$ for all $m>0$ so that $\Qo(X)\eta=\sum_{n=0}^\infty c_n e_{-n}$ for some square summable sequence of complex numbers $c_n$, i.e.\ $\Qo(X)\eta$ is a Hardy function which vanishes on a set $\T\setminus X$ of measure $1-\ell(X)>0$.
As well known, a Hardy function
which vanishes on a set of positive measure is identically zero (see, e.g., \cite[Theorem 1]{Helson}). 
Therefore, 
  $\Qo(X)\eta=0$, $\eta'=0$, and $\alpha\kb{e_0}{e_0}=0$, yielding $\alpha=0.$ 
\end{proof}

\begin{lemma}\label{apulause}
Let $E\in\lh$ be a positive operator such that $\<n|E|n\>=0$ for all $n>r$ where $r\in\N$, and let $X\in\bo{\T}$ be such that $\Phi(X)\ne I$. Then $E\le\Phi(X)$ implies $E=0$.
\end{lemma}

\begin{proof}
The proof is by induction on $r$.
First we note that, by positivity, if $\<n|E|n\>=0$ for some $n$, then $\<m|E|n\>=\overline{\<n|E|m\>}=0$ for all $m\in\N$.
The condition $E\le\Phi(X)$ implies
$$
\<r|E|r\>\kb{0}{0}=WEW^*\le W\Phi(X)W^*=\Phi(X)
$$
where $W=\sum_{k=0}^\infty \kb{k}{{k+r}}$. From Lemma \ref{lemma} one gets $\<r|E|r\>=0$ and by induction
$\<n|E|n\>=0$ for all $n\in\N$, i.e.\ $E=0$.
\end{proof}

\begin{theorem}\label{corollary}
For any finite subset $Y$ of $\N$ and $X\in\bo{\T}$ such that $\Phi(X)\ne I$,
$$
 \Phi(X)\wedge\No(Y) =0.
$$
\end{theorem}

\begin{proof}
Clearly, the claim holds if $Y=\emptyset$ (i.e.\ $\No(Y)=0$) so that we assume that $Y$ is finite and non-empty.  
Assume that there is an effect $E$ such that $E\le \Phi(X)$ and $E\le \No(Y)$.
Thus, $r=\max Y\in\N$, $\No(Y)\le R=\sum_{n=0}^r\kb{n}{n}$, 
$\<n|E|n\>\le\<n|R|n\>=0$ for all $n>r$. Since also $E\leq\Phi(X)$, Lemma \ref{apulause} now implies that $E=0$, that is, 0 is the only lower bound of  $\Phi(X)$ and $\No(Y)$.
\end{proof}

We note that \eqref{komplementaarisuus} is equivalent with the seemingly weaker requirement that this condition holds for all singletons $Y=\{n\}$. Finally, we give bounds for  the joint predictability of number and phase. 

\begin{corollary}\label{Lenard}
For any $X\in\bo{\T}$, with $\ell(X)<1$, and for any  finite $Y\subset\N$,
$$
\sup_{\rho\in\mathcal S(\hi)}\left(\Phi_\rho(X)+\No_\rho(Y) \right) \leq 1+\sqrt{a_+}<2,
$$
where $a_+$ is the largest eigenvalue the (finite rank) operator $\No(Y)\Phi(X)\No(Y)$. 
\end{corollary}
\begin{proof}
Considering $\Phi$ and $\No$ as the Naimark 
projections of $\Qo$ and $\Po$ on the subspace  $P_\N(\widetilde{\hi})$ of $L^2(\T)$, we have
$$
\sup_{\rho\in\mathcal S(\hi)}\left(\Phi_\rho(X)+\No_\rho(Y) \right) \leq \sup_{\rho\in\mathcal S(\widetilde\hi)}\left(\Qo_\rho(X)+\Po_\rho(Y) \right).
$$
Using the results of \cite{Lenard1972} the numerical range $\{(\ip{f}{\Po(Y)f},\ip{f}{\Qo(X)f})\,|\, f\in\widetilde\hi, \no{f}=1\}$
of  the pair of projections $\Po(Y), \Qo(X)$ can completely be determined. Since $\Po(Y)\land\Qo(X)=0$, the point $(1,1)$ is now excluded from this range.
It suffice to recall here that the numerical range is a convex subset of $[0,1]\times[0,1]$ \cite[Proposition 1]{Lenard1972} 
and that for any unit vector  
$f\in L^2(\T)$, the sum
$\ip{f}{\Qo(X) f}+\ip{f}{\Po(Y)f}$ is bounded by the number $1+\sqrt{a_+}$, where $a_+$ is
the maximal eigenvalue of the positive finite rank operator $\Po(Y)\Qo(X)\Po(Y)$
\cite[Proposition 5]{Lenard1972}. Note that the spectra of the operators $\No(Y)\Phi(X)\No(Y)$ and $\Po(Y)\Qo(X)\Po(Y)$ are identical. Since $\tr{\rho\Phi(X)}<1$ for any state $\rho\in\mathcal S(\hi)$ (see, for instance, \cite[Proposition 16.2]{Kirja}), the eigenvalue $a_+$ is strictly less than one.
\end{proof}

\section{Errors in approximate joint measurements of $\No$ and $\Phi$}\label{MURchapt}
We study next  the necessary errors appearing in an approximate joint measurement of number and canonical phase. 
We follow the idea, expounded, for instance, in \cite[pp.\ 197-8]{Ludwig1983}, that ``measurement error" is to be found by comparing a ``real" measurement outcome statistics with the desired one. We take this to mean the comparison of the actual measurement outcome distributions with the ideal ones.
Such a comparison can be based on various methods. Here we follow the approach initiated in \cite{RFW2004} and further developed in \cite{BLWJMP2014,BLW2014}
where the error is quantified  using the Wasserstein distance between probability measures.
For simplicity, we use only the Wasserstein-2 distances and   fix the metrics to be the arc distance on $\T$, $d(\theta,\theta') = \min_{n\in\Z}\vert \theta - \theta' - 2\pi n\vert$, and the standard distance on $\N$, $d(m,n) = \vert m-n\vert$. 

Let  $\Mo_1:\bo{\T} \to\lh$ 
and $\Mo_2:\bo{\N}\to\lh$
be any two observables (semispectral measures) which  approximate  measurements of $\Phi$ and $\No$, respectively. 
The error in approximating $\Phi$ by $\Mo_1$ is now defined as
\begin{equation}
d(\Mo_1,\Phi) = \sup_{\rho}D((\Mo_1)_\rho,\Phi_\rho), 
\end{equation}
where $D((\Mo_1)_\rho,\Phi_\rho)$ is the Wasserstein-2 distance between the probability measures $(\Mo_1)_\rho$ and $\Phi_\rho$, that is,
$$
D((\Mo_1)_\rho,\Phi_\rho)=\inf_\gamma  \sqrt{\int_{\T\times\T} d(\theta,\theta')^2\, d\gamma(\theta,\theta')},
$$
where the infimum is taken over all couplings (joint probabilities) $\gamma:\bo{\T\times\T}\to[0,1]$ of $(\Mo_1)_\rho$ and $\Phi_\rho$. Similarly,  one defines the error $d(\Mo_2,\No)$. Actually, the existence of a minimizing  coupling is known \cite[Theorem 4.1]{Villani}.

\begin{remark}\label{kalibrointi}
Canonical phase $\Phi$ is not a spectral measure. Still, as pointed out above, it resembles a spectral measure in many respects. In particular, the notion of calibration error
$$
d^c(\Mo_1,\Phi)= \lim_{\epsilon\to 0}\sup\{D((\Mo_1)_{\rho},\delta_x)\,|\, D(\Phi_\rho,\delta_x)\leq\epsilon\}
$$
makes sense, along with all spectral measure observables,  also to canonical phase and one has $d^c(\Mo_1,\Phi)\leq d(\Mo_1,\Phi)$. Moreover, if $\Mo_1$ is a smearing of $\Phi$ in the sense of a convolution, that is, $\Mo_1=\mu*\Phi$ for a probability measure $\mu$, then 
$d^c(\Mo_1,\Phi)^2=d(\Mo_1,\Phi)^2=\int_{\T}\,d(\theta,0)^2\,d\mu = \int_{\T}\,\min_{n\in\Z}|\theta-2\pi n|^2\,d\mu =\int_{-\pi}^{\pi}\theta^2d\mu \equiv
\mu[2]$. Similarly, if $\Mo_2=\nu*\No$ for some probability measure $\nu$, then  $d^c(\Mo_2,\No)^2=d(\Mo_2,\No)^2
=\sum_{k=0}^\infty d(k,0)^2\nu(\{k\})=
\sum_{k=0}^\infty k^2\nu(\{k\})\equiv\nu[2]$  \cite[Lemmas 7, 11]{BLWJMP2014}.
\end{remark}
  
 For an approximate joint measurement of $\Phi$ and $\No$, the approximators  $\Mo_1$ and $\Mo_2$ must be compatible, that is, margins of a joint observable
 $\Mo:\bo{\T\times\N}\to\lh$.\footnote{See \cite[Theorem 11.1]{Kirja} for several alternative definitions.} 
 The basic problem is thus to  characterize the joint measurement error set 
\begin{equation}\label{eqn:MU}
\mathsf{MU}(\T\times \N)= \{ (d(\Mo_1, \Phi), d(\Mo_2, \No) ) \mid \Mo:\bo{\T\times\N} \to\lh \},
\end{equation}
where $\Mo_j$ are the cartesian margins of $\Mo$. Here we use the notation $\mathsf{MU}(\T\times \N)$ to indicate explicitly the value space of the approximate joint observables.

The incompatibility of $\Phi$ and $\No$ implies that
 the point $(0,0)$ is not in the set $\mathsf{MU}(\T\times \N)$.   On the other hand, if one of the errors is zero, then $\Mo$ is a smearing of the exact margin  $\Mo_1$ or $\Mo_2$. From the below Proposition \ref{MURtulos} we then conclude that if $d(\Mo_1,\Phi)=0$, that is, $\Mo_1=\Phi$, then $d(\Mo_2,\No)$ cannot be finite.
On the other hand, if $\Mo_2=\No$, then $\pi/\sqrt3\le d(\Mo_1,\Phi)\le\pi$ where the lower bound is attained with the kernel $p_k=\ell$, $k\in\N$, and the upper bound with 
$p_k=\delta_\beta$, $k\in\N$, where $\beta\in[0,2\pi)$.

The semigroup structure of the outcome space of the number measurements has thwarted our attempts to determine directly the set \eqref{eqn:MU}. However, we can still bound this set by enlarging the joint values set $\T\times\N$ to $\T\times\Z$, that is, studying instead of \eqref{eqn:MU} the set  $\mathsf{MU}(\T\times \Z)$. This case reduces to the case of position $\Qo$ and momentum $\Po$ (or angle and ($\Z$ -)number) on $\widetilde\hi=L^2(\T)$ studied in great detail  in
 \cite{BKWerner2016}.

 Let $\Go^\sigma:\bo{\T\times \Z}\to \mathcal{L}(\widetilde\hi)$ be the covariant phase space observable generated by a state $\sigma\in\mathcal S(\widetilde\hi)$ so that
 its margins are the smeared position and momentum observables $\Qo_\sigma *\Qo$ and $\Po_\sigma *\Po$, 
 smeared by the position and momentum distributions  $\Qo_\sigma$ and $\Po_\sigma$ in state $\sigma$,
 respectively  \cite{Werner1984,BKWerner2016}.
 The observable $\Eo^\sigma:\bo{\T\times \Z}\to \lh$, defined as
\begin{equation}\label{eqn:joint}
\Eo^\sigma(X\times Y) = V^* \Go^\sigma(X\times Y) V,
\end{equation}
has then the smeared phase $\Eo^\sigma_1 = \Qo_\sigma * \Phi$ and smeared number $\Eo^\sigma_2 = \Po_\sigma* \No$ as its margins. By Remark \ref{kalibrointi}, the errors now reduce to the preparation uncertainties  
of $\Qo$ and $\Po$ in state $\sigma$  
$$  
d(\Eo^\sigma_1, \Phi) =  \sqrt{\Qo_\sigma[2]}    
\quad {\rm and}\quad 
d(\Eo^\sigma_2, \No) =  \sqrt{\Po_\sigma[2]}. 
$$

The following proposition  bounds the error set  $\mathsf{MU}(\T\times \N)$ by the bounds of the larger set $\mathsf{MU}(\T\times \Z)$.
\begin{proposition}\label{MURtulos}
Let $\Fo:\bo{\T\times \Z}\to\lh$ be an observable such that $d(\Fo_2,\No) <\infty$. Then there exists a state operator  
$\sigma$ on  $\widetilde\hi$, such that 
$$
d(\Eo^\sigma_1, \Phi) \leq d(\Fo_1,\Phi) \qquad \text{and} \qquad d(\Eo^\sigma_2, \No) \leq d(\Fo_2,\No),
$$
where $\Eo^\sigma$ is given by \eqref{eqn:joint}. In particular, the boundary curve for the error set $\mathsf{MU}(\T\times \Z)$,  which includes the set  $\mathsf{MU}(\T\times \N)$, is the same as for  $\Qo$ and $\Po$ on $\widetilde\hi$, as characterised in \cite{BKWerner2016}.
\end{proposition}
\noindent
The idea behind the proof is the following:
\begin{enumerate}
\item Starting from $\Fo$, construct an observable $\Mo$ on $\widetilde\hi$ in such a way that the errors of its margins with respect to $\Qo$ and $\Po$ reflect the original errors.
\item Average $\Mo$ with respect to phase space translations so that the errors (actually, the state dependent errors) do not increase. 
\item Project the averaged observable $\overline{\Mo}$ back to $\hi$ to get the desired result. 
\end{enumerate}
\begin{proof}
Let $\Fo:\bo{\T\times \Z}\to\lh$ be an observable with $d(\Fo_2,\No) <\infty$. Define an observable $\Mo:\bo{\T\times \Z}\to \mathcal{L} (\widetilde\hi)$ via 
\begin{equation}\label{eqn:M}
\Mo(X\times Y) = V \Fo(X\times Y) V^* +  \sum_{n=1}^\infty \ell(X)\langle n \vert \Fo_2(-Y) \vert n \rangle \vert e_{-n}\rangle \langle e_{-n}\vert. 
\end{equation}

We now proceed by calculating the error $d(\Mo_2,\Po)$ for the second margin $\Mo_2$. 
By Remark \ref{kalibrointi}, it is sufficient to take the supremum over the eigenstates $\vert e_k\rangle$ of $\Po$, and we have the probabilities
$$
{\bf p}_{e_k}^{\Mo_2} (Y) = \langle e_k \vert\Mo (\T\times Y) e_k\rangle = \left\{ \begin{array}{ll}
\langle k \vert \Fo_2(Y) \vert k \rangle & \text{ for }k\geq 0,\\
\langle -k \vert \Fo_2(-Y) \vert -k \rangle & \text{ for }k < 0.
\end{array}\right.
$$
Since ${\bf p}_{e_k}^{\Po} = \delta_{k}$, we have 
$$
d({\bf p}_{e_k}^{\Mo_2}, {\bf p}_{e_k}^{\Po} ) = \left( \sum_{l=-\infty}^\infty \vert l-k\vert^2  \, {\bf p}_{e_k}^{\Mo_2}(\{ l \}) \right)^{1/2}
$$
so that for $k\geq 0$, 
$$
d({\bf p}_{e_k}^{\Mo_2}, {\bf p}_{e_k}^{\Po} ) = \left(   \sum_{l=-\infty}^\infty \vert l-k\vert^2  \,  \langle k \vert \Fo_2(\{ l\}) \vert k \rangle    \right)^{1/2}  = d({\bf p}_{\vert k\rangle}^{\Fo_2}, {\bf p}_{\vert k\rangle}^{\No} ) 
$$
whereas for $k<0$ we have
\begin{align*}
d({\bf p}_{e_k}^{\Mo_2}, {\bf p}_{e_k}^{\Po} ) &= \left(   \sum_{l=-\infty}^\infty \vert l-k\vert^2  \,  \langle - k \vert \Fo_2(\{ - l\}) \vert - k \rangle    \right)^{1/2}  \\
&= \left(   \sum_{l=-\infty}^\infty \vert l -(- k)\vert^2  \,  \langle - k \vert \Fo_2(\{  l\}) \vert - k \rangle    \right)^{1/2}  \\
&= d({\bf p}_{\vert -k\rangle}^{\Fo_2}, {\bf p}_{\vert -k\rangle}^{\No} ) 
\end{align*}
Since $d(\Fo_2,\No)$ is also obtained by calculating the supremum over the number states $\vert k\rangle$, we have that 
\begin{equation}\label{eqn:M_2}
d(\Mo_2,\Po) = \sup_{k\in\Z} d({\bf p}_{e_k}^{\Mo_2}, {\bf p}_{e_k}^{\Po} ) = \sup_{k\in\N}d({\bf p}_{\vert k\rangle}^{\Fo_2}, {\bf p}_{\vert k\rangle}^{\No} ) = d(\Fo_2,\No).
\end{equation}

For the first margin, we do not get such an equality due to the trivial term coming from the last term in Eq.~\eqref{eqn:M}. However, we may  restrict to the states 
$$
\mathcal{S}_+ (\widetilde\hi) = \{ \rho \in\mathcal{S}(\widetilde\hi)\mid \langle e_k\vert \rho e_l\rangle = 0 \text{ for all }k<0\text{ or }l<0    \}
$$
so that $V^* \mathcal{S}_+ (\widetilde\hi) V = \mathcal{S}(\hi)$. 
Since for any $\rho\in \mathcal{S}_+ (\widetilde\hi)$ we have $\tr{\rho \Mo_1( Y)} = \tr{V^*\rho V \Fo_1(X)}$ and $\tr{\rho\Qo(X)}  = \tr{VV^*\rho VV^*\Qo(X)} = \tr{V^*\rho V \Phi(X) }$, we have, in particular,  that 
\begin{equation}\label{eqn:M_1}
d(\Fo_1,\Phi)  =   \sup_{\rho\in\mathcal{S}(\hi)} d({\bf p}_\rho^{\Fo_1}, {\bf p}_\rho^\Phi) = \sup_{\rho\in\mathcal{S}_+(\widetilde\hi)} d({\bf p}_{V^*\rho V}^{\Fo_1}, {\bf p}_{V^*\rho V}^\Phi)  = \sup_{\rho\in\mathcal{S}_+(\widetilde\hi)} d({\bf p}_{\rho}^{\Mo_1}, {\bf p}_{\rho }^\Qo).     
\end{equation}

The next step is to average the observable $\Mo$ with respect to phase space translations, and to show that the averaged  observable $\overline{\Mo}$ satisfies 
\begin{equation}\label{eqn:condition}
 \sup_{\rho\in\mathcal{S}_+(\widetilde\hi)} d({\bf p}_{\rho}^{\overline{\Mo}_1}, {\bf p}_{\rho }^\Qo) =\sup_{\rho\in\mathcal{S}_+(\widetilde\hi)} d({\bf p}_{\rho}^{\Mo_1}, {\bf p}_{\rho }^\Qo)  \quad\text{ and } \quad  d(\overline{\Mo}_2,\Po) 
 = d(\Mo_2,\Po) 
\end{equation}
 We perform the averaging by using an invariant mean $m$ on $\T\times\Z$, see, for instance, 
 \cite{HeRo}.
 For any trace class operator $T\in\mathcal{T}(\widetilde\hi)$ and any bounded continuous function $f:\T\times \Z\to\C$, define 
 $$
\Theta [T, f](\theta, k) = \tr{T W(\theta,k)^* \Mo(f^{(\theta,k)})W(\theta,k) }
$$
where $W(\theta,k)$ are the Weyl operators and $f^{(\theta,k)}$ denotes the translate of $f$. Then $\Theta [T, f]:\T\times \Z\to\C$ is a bounded continuous function, and by standard arguments the formula 
$$
\tr{T\overline{\Mo}(f)} = m\left( \Theta [T, f] \right)
$$
determines a covariant phase space observable $\overline{\Mo}:\mathcal{B}(\T\times \Z)\to\mathcal{L}(\widetilde\hi)$ 
(since $d(\Mo_2,\Po) = d(\Fo_2,\No)<\infty$ and $d(\Mo_1,\Phi)<\infty $ trivially by the compactness of $\T$, the normalization of $\overline{\Mo}$ is guaranteed \cite{RFW2004}).

Let $\rho\in  \mathcal{S}(\widetilde\hi)$. Then by the Kantorovich duality, for any bounded continuous functions $f,g:\T\to\R$ such that $f(\theta)- g(\theta')\leq d(\theta,\theta')^2$ we have 
$$
\tr{\rho (\Mo_1(f) - \Qo(g))} \leq d({\bf p}_{\rho}^{\Mo_1}, {\bf p}_{\rho }^\Qo).
$$
Since the above class of functions is invariant with respect to translations, we have 
\begin{align*}
\tr{W(\theta,k)\rho W(\theta,k)^* (\Mo_1(f^{(\theta)}) - \Qo(g^{(\theta)}))} &=\tr{\rho W(\theta,k)^* \Mo_1(f^{(\theta)}) W(\theta, k)} - \tr{\rho \Qo(g)}  \\
&\leq d({\bf p}_{\rho}^{\Mo_1}, {\bf p}_{\rho }^\Qo), 
\end{align*}
or equivalently, 
$$
\tr{\rho W(\theta,k)^* \Mo(f_1^{(\theta,k)}) W(\theta, k)} \leq \tr{\rho \Qo(g)}  + d({\bf p}_{\rho}^{\Mo_1}, {\bf p}_{\rho }^\Qo)
$$
where $f_1(\alpha,l) = f(\alpha)$. By applying the invariant mean, we obtain
$$
\tr{\rho \overline{M}_1(f)}  - \tr{\rho \Qo(g)}  \leq d({\bf p}_{\rho}^{\Mo_1}, {\bf p}_{\rho }^\Qo)
$$
for all $f,g$. By taking the supremum over such functions we get
$$
d({\bf p}_{\rho}^{\overline{\Mo}_1}, {\bf p}_{\rho }^\Qo) \leq d({\bf p}_{\rho}^{\Mo_1}, {\bf p}_{\rho }^\Qo)
$$
for all $\rho\in  \mathcal{S}(\widetilde\hi)$. The same holds also for the second margin. Hence, we conclude that Eq.~\eqref{eqn:condition} holds.

Since $\overline{\Mo}$ is a covariant phase space observable, we know that $\overline{\Mo} = \Go^\sigma$ for some $\sigma\in  \mathcal{S}(\widetilde\hi)$. We now set $\Eo^\sigma = V^* \Go^\sigma V = V^* \overline{\Mo}V$, so that 
\begin{align*}
d(\Eo^\sigma_1,\Phi) &=  d(V^*\overline{\Mo}_1 V, V^*\Qo V)= \sup_{\rho \in\mathcal{S}(L^2(\R))} d({\bf p}_{V\rho V^*}^{\overline{\Mo}_1}, {\bf p}_{V\rho V^*}^{\Qo}) = \sup_{\rho \in\mathcal{S}_+(\widetilde\hi)} d({\bf p}_{\rho }^{\overline{\Mo}_1}, {\bf p}_{\rho }^{\Qo})  \\
&\leq \sup_{\rho \in\mathcal{S}_+(\widetilde\hi)} d({\bf p}_{\rho}^{\Mo_1}, {\bf p}_{\rho }^\Qo) = d(\Fo_1,\Phi)
\end{align*}
and similarly $d(\Eo^\sigma_2,\No) \leq d(\Fo_2, \No) $.

\end{proof}

For any $\Fo$ for which $d(\Fo_2,\No)$ is finite there is thus an $\Eo^\sigma$  such that
$d(\Eo^\sigma_1,\Phi) \leq d(\Fo_1,\Phi)$ and  $d(\Eo^\sigma_2,\No) \leq d(\Fo_2, \No)$, so that\footnote{Recall that due to the arc distance on $\T$, the error
$\Qo_\sigma[2]=\int_{-\pi}^\pi\theta^2\,d\Qo_\sigma(\theta)$ so that also the operator $Q^2=\int_{-\pi}^\pi\theta^2\,d\Qo(\theta)$.}
$$
d(\Fo_1,\Phi)^2+d(\Fo_2, \No)^2 \geq  d(\Eo^\sigma_1,\Phi)^2+d(\Eo^\sigma_2,\No) ^2= \Qo_\sigma[2]+\Po_\sigma[2]
= \tr{\sigma(Q^2+P^2)}\geq \widetilde E_0,
$$
where $\widetilde E_0>0$ is the smallest eigenvalue of the oscillator energy operator $Q^2+P^2$ in $\widetilde\hi$. Though the existence of $\widetilde E_0$ is known, we can only give its  approximate value $\widetilde E_0\approx 0.9996$ (see Appendix \ref{liite}). If $\psi\in\widetilde\hi$ is a corresponding eigenvector then $\Eo^{\kb\psi\psi}$ is an optimal joint measurement of $\Phi$ and $\No$ with the value space  $\T\times\Z$.
For a detailed analysis of the boundary curve of the convex hull of the monotone hull of the error sets $\mathsf{MU}(\T\times \Z)$ we refer to \cite{BKWerner2016}, in particular, its Sections IV, V, and VI.

\begin{remark}{\rm
By extending the value space of the approximate joint measurements from $\T\times \N$ to $\T\times\Z$, we are potentially enlarging also the initial error set.
This leaves us with a question if the inclusion 
$\mathsf{MU}(\T\times \N) \subseteq \mathsf{MU}(\T\times \Z)$
is a proper one.
Natural candidates for optimal joint observables on $\T\times \N$ are the observables $\Eo^\sigma$ whose support is contained in $\T\times \N$. This amounts to the requirement that the generating  operator $\sigma \in  \mathcal{S}(\widetilde\hi)$ is supported on the positive number states, that is, $\langle e_k| \sigma e_l\rangle = 0$ wherever $k<0$ or $l<0$. Optimizing over such states is equivalent to optimizing the preparation uncertainties for $\Phi$ and $\No$ over all states $\rho\in  \mathcal{S}(\hi)$. Based on numerical calculations, the uncertainties lead to a strict subset of $\mathsf{MU}(\T\times \Z)$ giving evidence that this inclusion could be a proper one. However, we are lacking an argument which would show that these are indeed optimal $\T\times \N$ valued approximate joint observables. We are thus also left with the problem of
proving or disproving that the optimal $\T\times \N$ valued approximate joint observables for $\Phi$ and $\No$ are given by those $\Eo^\sigma$ whose support is contained in $\T\times\N$.}
\end{remark}

\section*{Acknowledgments}
JS acknowledges financial support from the EU through the Collaborative Projects QuProCS (Grant Agreement No. 641277).

\appendix 
\section{Proof of the existence of the eigenvalue}\label{liite}
In this appendix we give a simple proof of the well-known fact that the operator $P^2+Q^2$ in $\widetilde\hi$, as well as the operator $N^2+\Phi[2]$ in $\hi$, has a discrete spectrum with a strictly positive lowest eigenvalue.
For that end, we fix a separable Hilbert space (with the identity $I$) and assume that all operators (bounded or not) act in this space. We let $\mathcal B$ denote the unit ball of the Hilbert space.

\begin{lemma}\label{lemma3}
Let $E$ and $F$ be bounded operators such that $0\le E\le F\le I$ and $\|I-E\|<1$. Then $E$ and $F$ are  invertible and $E^{-1}\ge F^{-1}\ge I$.
\end{lemma}

\begin{proof}
Since $\|I-E\|<1$
it follows that
$\lim_{s\to\infty}\|I-E\|^s=0$, and
$
I+\sum_{k=1}^\infty(I-E)^k
$
converges in the operator norm to a bounded operator. Moreover, $$
\underbrace{E}_{I-(I-E)}[I+\sum_{k=1}^{s-1}(I-E)^k]=I-(I-E)^s\to I
$$ when $s\to\infty$, so that
$$
E^{-1}=I+\sum_{k=1}^\infty(I-E)^k\ge I.
$$
Indeed, $(I-E)^k=\int_{0}^{\|I-E\|}x^k\d\Mo(x)\ge 0$, for all $k=1,2,\ldots$, where $\Mo$ is the spectral measure of $I-E\ge 0$.
Since $0\le I-F\le I-E$ it follows that $\|I-F\|=\sup_{\psi\in\mathcal B}\<\psi|(I-F)\psi\> \le\|I-E\|<1$, and (similarly as above) one sees that $F$ is invertible.
Let $F^{1/2}$ (resp.\ $F^{-1/2}$) be the square root operators of $F$ (resp.\ $F^{-1}\ge I\ge 0$).
Now $G=F^{-1/2}EF^{-1/2}\ge0$ is invertible with the inverse 
$G^{-1}=F^{1/2}E^{-1}F^{1/2}$ and the condition $E\le F$ is equivalent to $G\le I$.
Now $\|I-G\|<1$ since otherwise (i.e.\  if $\|I-G\|=1$) there would exist a sequence $\{\psi_n\}_{n=1}^\infty\subset\mathcal B$ of unit vectors such that 
$\lim_{n\to\infty}\<\psi_n|(I-G)\psi_n\>=1$, that is, $\|G^{1/2}\psi_n\|^2=\<\psi_n|G\psi_n\>\to0$, $n\to\infty$, and
thus $1=\|\psi_n\|=\|G^{-1/2}G^{1/2}\psi_n\|\le\|G^{-1/2}\|\,\|G^{1/2}\psi_n\|\to 0$ when $n\to\infty$.
Hence, by the above calculation,
$G^{-1}\ge I$ so that
$
E^{-1}=F^{-1/2}G^{-1}F^{-1/2}\ge F^{-1}.
$
\end{proof}

\begin{proposition}\label{alaraja}
Let $T$ be a positive (possibly unbounded) selfadjoint operator with a purely discrete non-degenerate spectrum. Assume that its eigenvalues $0\leq p_0<p_1<p_2<\ldots$ are such that $\sum_n (1+p_n)^{-1}<\infty$.
Let $V$ be a  positive bounded  operator.
Then the spectrum of $H=T+V$ is discrete.
The lowest eigenvalue of $H$ is zero if and only if $p_0=0$ and $V\phi_0=0$ where $\phi_0\ne 0$ is an eigenvector of $T$ related to the eigenvalue $p_0$.
\end{proposition}

\begin{proof}
If the Hilbert space is finite dimensional then the proof is trivial so we consider only an infinite dimensional case.
By assumption,
$T=\sum_{n=0}^\infty p_n\kb{\phi_n}{\phi_n}$ for an 
orthonormal basis $\{\phi_n\}$. The domain of $T$ is ${\cal D}=\left\{\sum_{n=0}^\infty c_n\phi_n\,\Big|\,\sum_{n=0}^\infty p_n^2|c_n|^2<\infty\right\}$.
Now
 $(T+I)^{-1}=\sum_{n=0}^\infty p'_n\kb{\phi_n}{\phi_n}$, with $p_n'= (1+p_n)^{-1}\in(0,1]$,
 is a positive trace class operator.
Define $W= T+\|V\|\, I+I$ on $\mathcal D$ so that
$$
W^{-1/2}=\sum_{n=0}^\infty\frac1{\sqrt{p_n+\|V\|+1}}\kb{\phi_n}{\phi_n}
$$
is a bounded operator with the norm $\|W^{-1/2}\|=\sup_n(p_n+\|V\|+1)^{-1/2}=(p_0+\|V\|+1)^{-1/2}$.
Let $A=T+I$ and $B=T+V+I$ be positive operators defined on $\mathcal D$. Since $V\le\|V\| I$ one gets
$0\le\<\psi|A\psi\>\le\<\psi|B\psi\>\le\<\psi|W\psi\>$, $\psi\in\mathcal V={\rm lin}\{\phi_n\}\subset\mathcal D$, or, since $W^{-1/2}\mathcal V\subset\mathcal V$, 
$$
0\le W^{-1/2}AW^{-1/2}\le W^{-1/2}BW^{-1/2}\le I
$$ 
where, e.g.\ $W^{-1/2}BW^{-1/2}$ is a bounded operator determined uniquely by the corresponding bounded sesquilinear form 
$\mathcal V\times\mathcal V\ni(\fii,\psi)\mapsto\<W^{-1/2}\fii|BW^{-1/2}\psi\>\in\C$.

Since $\|I-W^{-1/2}AW^{-1/2}\|=\sup_n\left(\frac{\|V\|}{p_n+\|V\|+1}\right)= \frac{\|V\|}{p_0+\|V\|+1}<1$, from Lemma \ref{lemma3}, one sees that 
$$
(W^{-1/2}AW^{-1/2})^{-1}\ge (W^{-1/2}BW^{-1/2})^{-1}\ge I,
$$
that is, $p_n'=\<\phi_n|(T+I)^{-1}\phi_n\>\ge\<\phi_n|(T+V+I)^{-1}\phi_n\>\ge (p_n+\|V\|+1)^{-1}>0$
and 
$$\sum_{n=0}^\infty\<\phi_n|W^{-1}\phi_n\>\le\sum_{n=0}^\infty\<\phi_n|(T+V+I)^{-1}\phi_n\>\le\sum_{n=0}^\infty p'_{n}<\infty$$
showing that
$(T+V+I)^{-1}\ge W^{-1}$ is a (positive) trace-class operator. 
Let  
$$
(T+V+I)^{-1}=\sum_{l=0}^\infty\lambda_l\kb{\fii_l}{\fii_l}
$$
where $\{\fii_l\}$ is an orthonormal basis and $\lambda_l\in(0,1]$, $\sum_{l=0}^\infty\lambda_l<\infty$.
Hence,
$$
H=T+V=\sum_{l=0}^\infty q_l\kb{\fii_l}{\fii_l}
$$ 
where $q_l=\lambda_l^{-1}-1\ge 0$. Finally, let $\phi\in\mathcal D$. Then, $H\phi=0$ if and only if $0=\<\phi|H\phi\>=\<\phi|T\phi\>+\<\phi|V\phi\>$ if and only if $\<\phi|T\phi\>=0=\<\phi|V\phi\>$ if and only if $T\phi=0=V\phi$.
\end{proof}

Note that, in the context of the above Proposition,
all operators $T+cV, c>0$, have discrete spectra, and their spectra have non-zero smallest eigenvalues (i.e.\ positive spectra) if $T+V$ has a positive spectrum.

In either case, $\widetilde{H}=P^2+Q^2=P^2+\int_{-\pi}^\pi\theta^2d\Qo(\theta)$ (in $\widetilde\hi$) or $H=N^2+\Phi[2]=N^2+\int_{-\pi}^\pi\theta^2d\Phi(\theta)$ (in $\hi$), the assumptions of Proposition \ref{alaraja} are satisfied; in particular, both of the positive operators $Q^2$ or $\Phi[2]$ have a purely continuous spectrum (with no eigenvalues): $\sigma(Q^2)=\sigma(\Phi[2])=[0,{\pi^2}]$. 
Hence both operators $\widetilde H, \, H$ have strictly positive lowest eigenvalues $\widetilde E_0, \, E_0$, respectively. 
Also, this follows directly from Proposition \ref{alaraja} by noting that
$\<e_0|Q^2e_0\>=\<0|\Phi[2]|0\>=\int_{-\pi}^{\pi}\theta^2d\theta/(2\pi)>0$, i.e.\ $P^2e_0=0$ but $Q^2e_0\ne0$ and $N^2|0\>=0$ but $\Phi[2]|0\>\ne 0$.
Numerically, $\widetilde E_0\approx0.9996...$ associated with the (normalized) eigenvector
$\tilde\psi_{\rm min}=\sum_{s=-\infty}^\infty c_{s}e_s$ where
$c_0\approx0.7518$, $c_{\pm1}\approx 0.4550$, $c_{\pm2}\approx0.1017$, $c_{\pm3}\approx0.0083$, $c_{\pm4}\approx0.0002$, etc.\
Moreover, $E_0\approx1.5818...$ with the eigenvector
$
\psi_{\rm min}\approx 0.7276 \ket0 +0.6632 \ket1 +0.1745 \ket2 +0. 0167\ket3 +0.0002 \ket4 +\ldots .
$
To conclude,
if $\Mo:\bo{\T\times\Z}\to\lh$ is any approximate joint measurement of $\Phi$ and $\No$, with $d(\Mo_2,\No)<\infty$, then
$$
 d(\Mo_1,\Phi)+d(\Mo_2,\No) \geq \widetilde E_0\approx 1.
$$
It remains, however,  an open question if the eigenvalue $E_0$ of $N^2+\Phi[2]$ bounds the error sum $d(\Mo_1,\Phi)+d(\Mo_2,\No)$ for the 
${\T\times\N}$-valued
approximate joint measurements  of phase and number.

\begin{remark}
The above numerical results for the smallest eigenvalues and the corresponding eigenvectors is based on the following facts:
Let $H=T+V$, $T=\sum_{n=0}^\infty p_n\kb{\phi_n}{\phi_n}$, be as in Proposition \ref{alaraja} (we assume that the Hilbert space is infinite-dimensional).
Let $C_{\psi_{\rm min}}\ge 0$ be the lowest eigenvalue of $H$ with the (normalized) eigenvector $\psi_{\rm min}$.
Let $P_k=\sum_{n=0}^k\kb{\phi_n}{\phi_n}$ so that $P_k\to I$, $k\to\infty$, with respect to the strong (and weak) operator topology. Denote $H_k=P_k H P_k\ge 0$ and let $\alpha_k$ be the smallest eigenvalue of the `finite positive matrix' $H_k$.
Let $\eta_k\in\mathcal B$, $P_k\eta_k=\eta_k$, be the corresponding eigenvector of $H_k$, that is, $H_k\eta_k=\alpha_k\eta_k$. 
Since $\alpha_k=\inf\{\<\psi|H_k\psi\>\,|\,\psi\in\mathcal B,\,P_k\psi=\psi\}$ and $P_{k+1}P_k=P_k$ one gets
$$
C_{\psi_{\rm min}}\le \<\eta_{k+1}|H\eta_{k+1}\>=\alpha_{k+1}\le\alpha_k\le
\<P_k\psi_{\rm min}|H_kP_k\psi_{\rm min}\>\|P_k\psi_{\rm min}\|^{-2}.
$$
Since $\lim_{k\to\infty}\|P_k\psi_{\rm min}\|=1$, to get $\lim_{k\to\infty}\alpha_k=C_{\psi_{\rm min}}$, 
one is left to show that (when $k\to\infty$)
$$
\<P_k\psi_{\rm min}|H_kP_k\psi_{\rm min}\>=\<P_k\psi_{\rm min}|HP_k\psi_{\rm min}\>\to\<\psi_{\rm min}|H\psi_{\rm min}\>=C_{\psi_{\rm min}}
$$
or\footnote{
$\<\psi|\cdots\psi\>=\<P_k\psi|\cdots P_k\psi\>+\<P_k^\perp\psi|\cdots P_k\psi\>+\<P_k\psi|\cdots P_k^\perp\psi\>+\<P_k^\perp\psi|\cdots P_k^\perp\psi\>$ where $P_k^\perp=I-P_k$
}
that
$
HP_k\psi_{\rm min}\to H\psi_{\rm min}=C_{\psi_{\rm min}}\psi_{\rm min}.
$
But this is obvious (see the end of the proof of the proposition):
$$
\|H\psi_{\rm min}-HP_k\psi_{\rm min}\|^2=\sum_{l=1}^\infty(q_l)^2\underbrace{|\<\fii_l|(I-P_k){\psi_{\rm min}}\>|^2}_{\to\;0\;(k\to\infty)}\to0.
$$
We have proved that $\lim_{k\to\infty}\alpha_k=C_{\psi_{\rm min}}$, i.e.\
$\lim_{k\to\infty} \<\eta_{k}|H\eta_{k}\>=\<\psi_{\rm min}|H\psi_{\rm min}\>$.
Hence, one can numerically solve the smallest eigenvalues $\alpha_k$ of the finite matrices $H_k$. When $k$ is large enough one gets  $C_{\psi_{\rm min}}\approx\alpha_k$.

\end{remark}

\end{document}